\documentclass{article}

\usepackage{amsmath,amsfonts,amssymb,amsthm}
\usepackage[english]{babel}
\usepackage{lscape}

\usepackage{tabulary}
\newcolumntype{K}[1]{>{\centering\arraybackslash}p{#1}}

\usepackage{mathtools}
\usepackage{enumerate}
\usepackage{cite}
\usepackage{hyperref}
\usepackage{comment}
\usepackage{todonotes}

\usepackage{multirow}
\usepackage[left=4 cm, right= 4 cm]{geometry}

\DeclareMathOperator{\Span}{Span}

\newcommand{\vF}{\mathbb{F}}

\newcommand{\vZ}{\mathbb{Z}}

\newcommand{\cC}{\mathcal{C}}

 \newtheorem{problem}{Problem}
 
 \newtheorem{proposition}{Proposition}
 \newtheorem{corollary}{Corollary}
 \theoremstyle{definition}

 \theoremstyle{remark}

\numberwithin{equation}{section}

\usepackage{afterpage}

\newcommand{\lweightt}[1]{\mathrm{wt}_{L_t}\left(#1\right)}

\newcommand{\ba}{\mathbf{a}}

\newcommand{\bg}{\mathbf{g}}
\newcommand{\bc}{\mathbf{c}}

\newcommand{\be}{\mathbf{e}}
\newcommand{\bq}{\mathbf{q}}
\newcommand{\br}{\mathbf{r}}

\newcommand{\by}{\mathbf{y}}
\newcommand{\bx}{\mathbf{x}}
\newcommand{\bu}{\mathbf{u}}
\newcommand{\bv}{\mathbf{v}}

\newcommand{\bA}{\mathbf{A}}
\newcommand{\bC}{\mathbf{C}}
\newcommand{\bE}{\mathbf{E}}
\newcommand{\bG}{\mathbf{G}}

\newcommand{\C}{\mathcal{C}}
\newcommand{\R}{\mathcal{R}}

\newcommand{\supp}{\mathrm{Supp}}
\newcommand{\proj}{\mathrm{Proj}}
\newcommand{\rank}{\mathrm{rank}}
\newcommand{\Enc}{\mathrm{Enc}}

\providecommand*\email[1]{\href{mailto:#1}{#1}}

 \title{On single server private information retrieval in a coding theory perspective}

 \author{Gianira N. Alfarano\thanks{University of Zurich, Institute of Mathematics, Zurich, Switzerland 
(\email{gianiranicoletta.alfarano@math.uzh.ch},\email{karan.khathuria@math.uzh.ch},\email{violetta.weger@math.uzh.ch})
 } \and Karan Khathuria\footnotemark[1] \and Violetta Weger\footnotemark[1]
 }

\begin{document}

\maketitle

\begin{abstract}
    In this paper, we present a new perspective of single server private information retrieval (PIR) schemes by using the notion of linear error-correcting codes. Many of the known single server schemes are based on taking linear combinations between database elements and the query elements. Using the theory of linear codes, we develop a generic framework that formalizes all such PIR schemes. Further, we describe some known PIR schemes with respect to this code-based framework, and present the weaknesses of the broken PIR schemes in a generic point of view. 
\end{abstract}

\section{Introduction}\label{sec:intro}

Private information retrieval (PIR) was first introduced in \cite{chor1995private} to cope with the following problem: retrieving an element from a database, without revealing to the untrusted source managing the database any information about that element. Since its introduction, it has attracted many researchers and several works have addressed their focus on it. There have been proposed two solutions to this problem, namely, the \emph{information theoretical} one and the \emph{computational} one. The first one aims to guarantee that the server gets no information about the file that the user wants to retrieve. Solutions for multiple servers were presented in \cite{dvir20162, beimel2002, sun2018capacity, sun2017capacity, banawan2018capacity, freij2017private}. In the case of a single server, the trivial solution, i.e., downloading the whole database, is the only possibility to ensure information theoretical privacy. On the contrary, in computational PIR, the privacy is guaranteed assuming that the server has limited computational power. Hence, the computational PIR (cPIR) can be used also in the case of a single server. 

Most of the early cPIR schemes are based on the difficulty of number-theoretical problems, such as integer factorization (see for example \cite{dong2014fast, QR, lipmaa2017simpler, stern1998new}).
The known (non-trivial) single server cPIR constructions require to perform some cryptographic operations on each database element, which increase the computational cost of these schemes in comparison to the information theoretical ones. 
In \cite{sion}, Sion and Carbunar showed that the number-theoretical PIR schemes are not practical, and computing a PIR reply is always less efficient than sending the whole database. Moreover, such schemes, based on factoring an integer, will be insecure in the era of quantum computers \cite{shor}. 

Some recent constructions of PIR schemes use a fully homomorphic encryption (FHE) scheme. Yi \emph{et al.} presented in  \cite{FHEtoPIR} a generic way to construct a PIR from an FHE. Following this construction many PIR protocols have been proposed using FHE schemes based on problems in lattices and learning with error (LWE) problems \cite{brakerski2014efficient,XPIR,SealPIR,ali2019communication}. Recently, Aguilar-Melchor \emph{et al.} presented in \cite{XPIR}  XPIR, a PIR construction using a Ring-LWE based FHE scheme, that is computationally efficient but comes with a large communication cost. Following \cite{XPIR}, Angel \emph{et al.} in \cite{SealPIR} were able to significantly improve its communication cost with only slightly more computations compared to XPIR. Along with the scheme of Angel \emph{et al.}, the recent work of Ali \emph{et al.} \cite{ali2019communication} represent the state-of-the-art efficiency for PIR schemes. 

Recently,  Holzbaur,  Hollanti  and  Wachter-Zeh  have  proposed in \cite{holzbaur2020computational} the first single server PIR based on coding theory. However, their proposal was attacked in \cite{bordage2020privacy}. The primary idea in \cite{holzbaur2020computational} is to generate the query by hiding carefully chosen error vectors using codewords from a random linear code. The linear code is kept secret by the user in order to obtain privacy. The same idea was previously used by Aguilar-Melchor and Gaborit in a lattice-based PIR scheme \cite{amg}, without using the notion of linear codes. The scheme was later attacked by Liu and Bi \cite{liu2016cryptanalysis} using lattice reduction algorithms. 

Interestingly, the idea of hiding query information using linear codes can be observed, directly or indirectly, in several other PIR schemes. In this paper, we develop a unified framework that describes all such PIR schemes. In particular, this framework characterizes all the single server PIR schemes that generate replies by contracting the database elements and the query elements using linear combinations. The main aim of this paper is to bring together and analyze several existing single server PIR schemes in a coding theoretic perspective. 

The framework is based on two key elements: a linear code that hides the query information, and a retrieval function that allows the user to retrieve the desired file from a linearly entangled  reply. On one hand, the notion of linear codes describes the common features of several existing PIR schemes, and on the other hand, the retrieval function describes the key differences between the schemes. In terms of the framework, the privacy of a PIR scheme heavily relies on the retrieval function. We observe that several choices of retrieval functions are not safe to use, for example, finite field homomorphisms and vector space homomorphisms. Moreover, we discuss the weaknesses of many broken PIR schemes with respect to this code-based framework. 

The paper is organized as follows: 
in Section \ref{sec:preliminaries}, we introduce the notation that will be used throughout the paper and give the background on single server private information retrieval, and linear codes over finite fields and over rings. In Section \ref{sec:general_framework}, we present the code-based framework and discuss the security in a general point of view. In Section \ref{sec:examples}, we provide four different examples of PIR schemes, described in terms of the code-based framework. The first example is a basic scheme that uses a finite field homomorphism as the retrieval function. The rest of the examples are based on the existing PIR schemes \cite{holzbaur2020computational}, \cite{amg} and \cite{XPIR}, respectively. Finally, in Section \ref{sec:remarks}, we draw some theoretical remarks on the generality of the framework, and on the security of single server PIR schemes.


\section{Preliminaries}\label{sec:preliminaries}

In this section we  introduce the notation that we use in the paper and we recall some background on the theory of single server PIR.
Moreover, we introduce the basic notions of error-correcting linear codes. 

\subsection{Notation}
In this paper, we denote by $\mathcal{R}$ a ring and by $\mathcal{R}^\times$ the set of invertible elements in the ring $\mathcal{R}$. Moreover, let $q$ be a prime power, then we denote by $\vF_q$ the finite field of size $q$. 

We use bold lower case, respectively bold upper case letters to denote row vectors, respectively matrices. When we consider column vectors, we use   the transpose symbol. The identity matrix of size $k$ is denoted by $\mathbf{I}_k$.
Given a vector $\mathbf{x}$ of length $n$ and a set $S \subset \{1, \ldots, n\}$,  we denote by $\mathbf{x}_S$ the projection of $\mathbf{x}$ on the coordinates indexed by $S$.
In the same way, $\mathbf{M}_{S}$ denotes the projection of the $k\times n$ matrix $\mathbf{M}$ to the columns indexed by $S$.

For a set $S$ we denote by $S^C$  its complement. The support of a vector $\mathbf{x}\in \mathbb{F}_q^n$ is denoted by  $\text{Supp}(\mathbf{x})=\{ 1 \leq i \leq n \mid x_i \neq 0\}.$

The $i$-th entry of a vector $\mathbf{x} \in \mathbb{F}_q^n$ is denoted by $\mathbf{x}[i]$, for $i \in \{1, \ldots, n\}$.

Given a set $S$ and a distribution $\chi$ on $S$, $x \leftarrow \chi$ represents a sample $x$ from $S$ following the distribution $\chi$.

\subsection{Single server private information retrieval}
A single server PIR is a scheme involving two parties, the \emph{user} and the \emph{server}. The server manages a database containing some public information, and the user is interested in retrieving some entries of the database, without revealing which item was queried. 

\subsubsection{Basic description}
A basic description of a single server PIR scheme is as follows. Let the database be denoted by $\mathcal{DB}=\{db_1,\ldots,db_N\}$, containing $N$ files, and suppose the user wishes to retrieve the $i$-th file $db_i$. The user first constructs a \textit{query} $Q = \{q_1,\ldots,q_N\}$, which hides the information about the index $i$, and sends it to the server. The server computes a \textit{response} by performing certain operations between $q_j$ and $db_j$ for each $j$, and returns it to the user. The scheme is said to be \emph{correct} if the user can retrieve the desired file $db_i$ from the response.

\subsubsection{Communication and computational cost}
A simple solution to preserve the privacy is downloading the whole database. However, the communication cost of this operation, measured as the total number of bits exchanged by user and server, in the trivial case is too high, namely $\mathcal{O}(N)$ where $N$ is the size of the database.
Modern PIR protocols allow the user to retrieve data from the database, with a communication complexity much smaller than $\mathcal{O}(N)$. Some common methods can be used to improve the communication cost of any PIR scheme. In Section \ref{sec:techniques}, we discuss such techniques in detail.

Another important aspect of a single server PIR scheme is the computational cost. Since the database has to process each entry of the query, the schemes are computationally expensive.

\subsection{Linear codes}
\subsubsection{Over finite fields}
Let $\bx$ be a vector in $\vF_q^n$. The \emph{Hamming weight} of $\bx$ is denoted by $\mathrm{wt}(\bx)$ and it is defined as the number of its nonzero entries, i.e., it is the size of its support. The \emph{Hamming distance} between two vectors $\bx,\by \in \vF_q^n$ is defined as the number of components in which the two vectors differ, i.e.,  $d(\bx,\by) = |\{i \ \mid \ x_i\ne y_i\}|$. 

 An $[n,k]_q$ \emph{linear code} $\mathcal{C}$ is a $k$-dimensional subspace of $\vF_q^n$ endowed with the Hamming distance and the elements of $\mathcal{C}$ are called \emph{codewords}. 

 The \emph{minimum distance} $d$ of $\mathcal{C}$ is the quantity
 $$d:=\min\{d(\mathbf{x},\mathbf{y})\mid \mathbf{x},\mathbf{y} \in \mathcal{C}, \mathbf{x} \neq \mathbf{y} \}.$$
 When the minimum distance $d$ of a linear code $\mathcal{C}$  is known, then $\mathcal{C}$ is denoted by $[n,k,d]_q$. 
 
 A matrix $\mathbf{G}\in\vF_q^{k\times n}$ whose rows form a basis for $\mathcal{C}$ is called \emph{generator matrix} of $\mathcal{C}$. Hence, we can define the code $\mathcal{C}$ as $\{\bv \in\vF_q^n\mid \bv=\bu \mathbf{G}^\top, \bu\in\vF_q^k\}$. Similarly, we can define the code $\cC$ as the kernel of a matrix $\mathbf{H}\in\vF_q^{(n-k)\times n}$, i.e. $\mathcal{C}:=\ker(\mathbf{H}) = \{\bv\in\vF_q^n\mid \mathbf{H}\bv^\top = \mathbf{0}\}$. Such a matrix is called \emph{parity-check matrix} for the code $\mathcal{C}$. An information set of an $[n,k,d]_q$ code $\mathcal{C}$ is a set $I \subset \{1, \ldots, n\}$ of size $k$, such that $\mid \mathcal{C} \mid = \mid \mathcal{C}_I \mid,$ where $\mathcal{C}_I$ denotes the restriction of all codewords to the entries indexed by $I$.

\subsubsection{Over rings}
Let $\R$ be a commutative ring with identity. A \textit{linear code} $\cC$ of length $n$ over $\R$ is an $\R$-module in the space $\R^n$. 

A linear code $\cC$ of length $n$ over $\R$ is called \textit{cyclic} if $\bc = (c_1,\ldots,c_n) \in \cC$ implies $(c_{n},c_1,\ldots,c_{n-1}) \in \cC$. Equivalently, $\cC$ is an ideal of the ring $\R[x]/(x^n-1)$. 

A linear code $\cC$ of length $n$ over $\R$ is called \textit{constacyclic} if $\cC$ is an ideal of the ring $\R[x]/(x^n+1)$.



\section{Code-based framework}\label{sec:general_framework}

In this section, we present a generic framework for single server PIR schemes by using the notion of error-correcting codes. For simplicity, we present the framework using a simple database setup, later we discuss different kinds of database setups that can be used to improve the communication complexity.
\subsection{Code-based framework}

Before we describe the framework in detail, we highlight some elements that are used in the framework:
\begin{itemize}
\item We describe the generic framework over a finite commutative ring $\R$ using a retrieval function $f:\R \to \R$ and three subsets $X,Y,Z$ of $\R$. 
\item The database files belong to the set $X$.
\item In order to generate queries, we fix a randomly chosen linear code $\cC$ over $\R$. Each element of the query is the sum of a randomly chosen codeword in $\cC$ and an error vector over $\R$.
\item To generate the error vectors corresponding to the non-desired files we use the set $Y$, whereas for the desired file we use the set $Z$.
\end{itemize}

\paragraph{Setup:}

We define a retrieval function $f: \R \to \R$, and subsets $X,Y,Z \subseteq \R$ satisfying:
\begin{enumerate}
    \item $f$ is a non-zero map. 
    \item $Y \subseteq \ker(f):= \{x \in \R: f(x)=0\}$ such that any linear combination of elements in $Y$ with scalars in $X$ belongs to $\ker(f)$, i.e., $x_1y_1+x_2y_2+ \cdots + x_jy_j \in \ker(f)$ whenever $x_1,\ldots,x_j \in X$ and $y_1,\ldots,y_j \in Y$.
    \item $Z \subseteq f^{-1}(\R^\times)$ such that $f(y+xz) = xf(z)$ for all $y \in \ker(f), x \in X$ and $z \in Z$.
\end{enumerate}

Note that  $f$ does not need  to be a ring homomorphism, it can be any kind of function from $\R$ to $\R$ satisfying the above three conditions. 

Let $\mathbf{M} = (m_{i}) \in X^{N}$ represents the database, i.e., there are $N$ files in the database. Suppose that the user wants to retrieve the $b$-th file from the database. \\
Let $\cC$ be a random linear code over $\R$ of length $n$, i.e., $\C$ is an $\R$-submodule of $\R^n$.

\paragraph{Query generation:}
Let $\bg_1,\ldots,\bg_m$ be generators of $\C$ as an $\R$-module, and let $\Enc: \R^m \to \R^n$ be an encoding map of $\C$. Note that $\Enc$ is an $\R$-linear map given by $(a_1,\ldots,a_m) \mapsto a_1 \bg_1 + \cdots + a_m \bg_m$.\\
Let $\ba_1,\ba_2,\ldots,\ba_N$ be randomly chosen elements in $\R^m$, and define $\bc_i = \Enc(\ba_i)$ for all $i \in \{1,\ldots,N\}$. 

Now, let $v$ be a randomly chosen fixed element in $\{1,\ldots,n\}$ and we randomly choose error vectors $\be_1,\be_2,\ldots,\be_N$ in $\R^n$, such that they satisfy the following conditions that allow the reply extraction:
\[ \be_b[v] \in Z \quad \mbox{and} \quad \be_i[v] \in Y \mbox{ for all } i \neq b. \]

Let $\bq_i:=(\ba_i,\bc_i+\be_i)$ for all $i\in \{1,\dots, N\}$.
The query is then given by 
\[Q := \{ \bq_1, \bq_2,\ldots,\bq_N \}.\]

\paragraph{Reply generation:} 
The response is generated by 
computing 
\[\br = \sum_{i=1}^N m_{i} \bq_i = \sum_{i=1}^N (m_{i} \ba_i, m_{i} (\bc_i+\be_i)) =: (\br_1,\br_2).\]

\paragraph{Reply extraction:}

First we perform the decoding by applying the encoding map on $\br_1$, and obtain:
\[\br_2 - \Enc(\br_1) = \sum_{i=1}^N m_i \be_i.\]
 After that we can use the retrieval function $f$ on the $v$-th coordinate, 
 \begin{align*}
     f\left(\sum_{i=1}^N m_i \be_i[v] \right) & = f\left(\sum_{i\neq b} m_i \be_i[v] \right) + f(m_b \be_b[v])  \\
     &= m_b f(\be_b[v]).
 \end{align*} 
The above equalities follow from the conditions of the retrieval function.
Now, since we know $f(\be_b[v])$ and we have that $f(\be_b[v]) \in f(Z) \subseteq \R^\times$,  we can retrieve the desired file $m_b$.

\subsection{Communication complexity and different database setups} \label{sec:techniques}
With respect to the basic description of the code-based framework, the communication cost is more than the size of the whole database. Indeed, for each file which is an element in $\R$, we are sending a query element in $\R^{n+m}$. Thus the total communication cost is $(N+1)$ times the size of an element in $\R^{m+n}$. We can improve the communication complexity by using a  matrix database setup \cite{chor1995private} or iterative response techniques: 
\begin{itemize}
    \item \textit{Matrix setup of database}: In order to reduce the communication complexity, one can see the database as an $s \times t$ matrix, where each element of the matrix is a database file. Now, the user generates a \emph{query} $Q = \{\bq_1,\dots, \bq_t\}$ containing $t$ elements. For each query, the server replies by sending back the \emph{response} $R=\{\br_1,\dots, \br_s\}$, which contains $s$ responses corresponding to the $s$ rows of the database matrix. This technique was introduced in \cite{chor1995private}. 
    Using this approach and assuming $s=t = \sqrt{N}$, the communication complexity is $2\sqrt{N}$ times the size of an element in $\R^{m+n}$.
    \item \textit{Iterative reply generation:} In this technique, one splits each file into $L$ parts and repeats the query to retrieve each part of the file. Since the query is generated in order to retrieve only small portions of the desired file, the size of the ambient space reduces accordingly. Hence, relative to the size of the database, the query size reduces by a factor of $L$, and the response size increases by the same factor. 
\end{itemize}

\subsection{Security} \label{sec:general-security}

The security of a single server computational PIR scheme is based on the difficulty of identifying the index of the desired file by looking at the query. With respect to the code-based framework, we can describe the security using the following distinguishability problem.

\begin{problem}[Distinguishability Problem]\label{problem}
Consider the notations of the setup and the query generation process of the code-based framework. Given the query vectors $\bq_1,\bq_2,\ldots,$ $\bq_N$, determine the index $b$ of the desired file.
\end{problem}

The difficulty of solving the distinguishability problem depends highly on the choice of the retrieval function $f$. In the following, we present two generic strategies that can be used to solve this problem. However, the computational cost of these strategies directly relies on the choice of the retrieval function and the error vectors $\be_1,\ldots,\be_N$.
\begin{enumerate}
    \item Consider the following matrix  consisting of the query vectors 
$$\mathbf{A} = \begin{pmatrix} \mathbf{q}_1 \\ \mathbf{q}_2 \\ \vdots \\ \mathbf{q}_N \end{pmatrix} = \begin{pmatrix} \mathbf{a}_1 & \bc_1+\be_1 \\ \mathbf{a}_2 & \bc_2 + \be_2 \\ \vdots \\ \mathbf{a}_N & \bc_N + \be_N \end{pmatrix} \in \R^{N \times (m+n)}.$$
    Observe that the vectors $\left(\be_1[j],\be_2[j],\ldots,\be_N[j]\right)$ for all $j \in \{1,\ldots,n\}$ belong to the column span of $\bA$. We recall that the $v$-th coordinate of the error vectors are chosen in a special way, i.e., $\be_b[v] \in Z \subseteq f^{-1}(\R^\times)$ and $\be_i[v] \in Y \subseteq \ker(f)$ for all $i \neq b$. Hence, one could solve Problem \ref{problem} by finding the vector $\left(\be_1[v],\be_2[v],\ldots,\be_N[v]\right)$ in the column span of $\bA$. 
    
    \item Let $\bA$ be the query matrix as defined above. For each $i \in \{1,\ldots,N\}$, let $\bA_i$ be the submatrix of $\bA$ obtained by deleting the $i$-th row. Clearly, by construction, $\bA_b$ has distinct properties compared to $\bA_i$ for any $i\neq b$. Thus, if there exists an (algebraic or non-algebraic) invariant that can distinguish $\bA_b$ from $\bA_i$ for any $i \neq b$, then Problem \ref{problem} can be solved by computing this invariant for each $\bA_1,\ldots,\bA_N$. 
\end{enumerate}

\section{Examples of different PIR's in our framework} \label{sec:examples}
In this section, we discuss several examples of single server PIR schemes that are based on different kinds of retrieval function. In each case, we analyze the security with respect to the distinguishability problem.  In Table \ref{comparison}, we summarize all the differences among the schemes.

\subsection{Basic PIR scheme using finite field isomorphism} \label{sec:pirate} 
In the following we describe the simplest case of the code-based framework, i.e., by considering linear codes over an arbitrary finite field and a field homomorphism for the retrieval function.

\subsubsection{Scheme}
\paragraph{Setup:}
Since the identity map is the only non-zero field endomorphism, the retrieval function $f: \vF_q \to \vF_q$ has to  be the identity map. We consider the sets $X = \vF_q$, $Y = \ker(f) = \{ 0\}$ and $Z= f^{-1}(\vF_q^\times) = \vF_q^\times$. It is easy to see that $f$ satisfies all the conditions of a retrieval function.\\

Let $\mathbf{M} = (m_i) \in \vF_{q}^{N}$ represents the database, i.e., there are $N$ files in the database, each file is of size $q$. Let $\cC$ be a random linear $[n,k]$ code over $\vF_q$.
The code $\mathcal{C}$ is kept secret by the user.

\paragraph{Query generation:}
Let $\bG$ be a generator matrix of $\C$, and let $I \subseteq \{1,\ldots,n\}$ be an information set. We use $\bG$ to perform the encoding, i.e., the encoding map  $\Enc: \vF_q^k \to \vF_q^n$ is given by $\ba \mapsto \ba \bG$.

Let $\ba_1,\ldots,\ba_N$ be randomly chosen vectors in $\vF_q^k$, and define the corresponding codewords
$\bc_i:= \Enc(\ba_i) = \ba_i \bG$  for all $i \in \{1, \ldots, N\}$.

Note that since $I$ is an information set, we have $(\bc_i)_I = \ba_i \bG_I$ for all $i \in \{1,\ldots,N\}$. Recall that in the code-based framework we send $\ba_i$'s in the query to facilitate the decoding in the reply extraction process. However, in this case, this can equivalently be achieved by adding no errors at the coordinates that are indexed by $I$. In particular, let $v$ be a random element in $I^C$, and we randomly choose error vectors $\be_1,\be_2,\ldots,\be_N$ in $\vF_q^n$ such that
\begin{enumerate}
    \item $\supp(\be_i) \subseteq I^C$ for all $i \in \{1, \ldots, N\}$,
    \item $\be_i[v] = 0$ for all $i \neq b$, and $\be_b[v] \neq 0$.
\end{enumerate}
Let $\bq_i:=\bc_i+\be_i$ for all $i \in \{1,\dots, N\}$.
The query is then given by 
\[Q := \{ \bq_1, \bq_2,\ldots,\bq_N\}.\]

\paragraph{Reply generation:} 

The database computes $$\br= \sum_{i=1}^N m_i \bq_i \in \mathbb{F}_q^n.$$

\paragraph{Reply extraction:} Write $\br = \bc + \be$, where $\bc := \sum_{i=1}^N m_i \bc_i$ and $\be := \sum_{i=1}^N m_i \be_i$.\\

Since $I$ is an information set and $\supp(\be) \subseteq I^C$, we can perform decoding on $\br$ by computing
$$\br - \br_I \bG_I^{-1} \bG = \be = \sum_{i=1}^N m_i \be_i.$$
We now only consider the $v$-th coordinate of $\be$ and apply the identity retrieval function, which gives $m_b\be_b[v]$, as for all $i \neq b$ we have that $\be_i[v] =0.$
Since $\be_b[v] \neq 0$, we can retrieve $m_b$.

\subsubsection{Security}

As we discussed in Section \ref{sec:general-security}, the security of the presented PIR scheme relies on the hardness of solving the distinguishability problem (see Problem \ref{problem}). In this case, the distinguishability problem can be solved in polynomial time using the  first strategy mentioned in Section \ref{sec:general-security}. 

Let $\bA$ be the matrix containing all the query vectors as rows, i.e., $$\mathbf{A} = \begin{pmatrix} \mathbf{q}_1 \\ \mathbf{q}_2 \\ \vdots \\ \mathbf{q}_N \end{pmatrix} = \begin{pmatrix} \bc_1+\be_1 \\   \bc_2 + \be_2 \\ \vdots \\  \bc_N + \be_N \end{pmatrix} = \bC + \bE,$$ with $\bC = \begin{pmatrix} \bc_1 \\ \vdots \\ \bc_N \end{pmatrix}$ and $\bE = \begin{pmatrix} \be_1 \\ \vdots \\ \be_N \end{pmatrix}$.
Since $I$ is an information set, we have
\begin{align*}
    \bA_I &= \bC_I + \bE_I = \bC_I, \\
    \bA_{I^C} & = \bC_{I^C} + \bE_{I^C} \\
    & = \bC_I \bG_I^{-1} \bG_{I^C} + \bE_{I^C} \\
    & = \bA_{I} \bG_I^{-1} \bG_{I^C} + \bE_{I^C}.
\end{align*} 
This implies that $$\bE_{I^C} = \bA_{I^C}  - \bA_I \bG_I^{-1} \bG_{I^C} ,$$ and hence the vector $\left(\be_1[v],\be_2[v],\ldots,\be_N[v]\right)$ belongs to the column span of $\bA$. We recall that $\be_b[v] \neq 0$ and $\be_i[v] = 0$ for all $i \neq b$. This means that the $b$-th unitary vector, i.e., the all zero vector  
having  the entry 1 at the $b$-th position, is in the column span of $\bA$. 

An attacker can easily find such a vector by simply going through all  $N$ unitary vectors and checking their existence in the column span of $\bA$. Moreover, existence of another vector of Hamming weight one in the column span of $\bA$ is very unlikely. More precisely, given an $N \times (n-1)$ random matrix $\bA$ where $N>n$, the probability of having a weight one vector in the column span of $\bA$ is $(n-1)q^{(n-N)}$, which is negligible. 
Despite having a small probability, there exist at most $n$ unit vectors in the column span of $\bA$, which leaks information about the index $b$, since $n <N$.

\subsection{HHWZ PIR scheme} \label{sec:lukas} Recently, Holzbaur, Hollanti and  Wachter-Zeh have proposed the first single server PIR scheme based on coding theory in \cite{holzbaur2020computational}.  In this PIR scheme the authors consider the field extension $\mathbb{F}_{q^m}$ and secretly choose a partition of the basis over $\mathbb{F}_q$. Shortly after, this proposal has been attacked in \cite{bordage2020privacy}, using that the removal of one row within the query matrix and checking for the dimension of the rest reveals the position of the desired file.

In the following, we describe this PIR scheme presented in \cite{holzbaur2020computational} with respect to our code-based framework. Later, we also present the attack \cite{bordage2020privacy} in terms of solving the distinguishability problem. 

Note that the original PIR scheme differs from our description in the database and query setup. In \cite{holzbaur2020computational}, the authors consider the database elements to be $L \times \delta$ matrices over the base field $\vF_q$, and the query elements are also $\delta \times n$ matrices over the base field $\vF_{q^m}$. Note that the authors have used the technique of iterative reply generation, i.e., by using the same query to retrieve each  of the $L$ rows of the database file. In the following description, we consider $L=1$ and use an equivalent setup where the database files are single elements in $\vF_{q}$ and the query elements are vectors over $\vF_{q^m}$.

\subsubsection{Scheme}
In this case, we work over an extension of the finite field $\vF_q$ and the retrieval function is an $\vF_q$-linear map. 

\paragraph{Setup:}
Let $\{\beta_1,\ldots,\beta_m\}$ be a basis of $\vF_{q^m}$ as an $\vF_q$-vector space. Further, let $\mathcal{V}$ be the subspace $\Span_{\vF_q}(\beta_1,\ldots,\beta_s)$ and  $\mathcal{W}$ be the subspace $\Span_{\vF_q}(\beta_{s+1},\ldots,\beta_m)$, where $s$ is some integer in $\{1,\ldots,m\}$. The retrieval function is given as \begin{align*}
    \proj_{\mathcal{V}}:  \vF_{q^m}&  \to \vF_{q^m}, \\ \sum_{i=1}^m \lambda_i \beta_i &\mapsto \sum_{i=1}^s \lambda_i \beta_i.
\end{align*}
Let $X$ be the set $\vF_q$, $Y=\ker(f) = \mathcal{W}$ and $Z= f^{-1}(\vF_{q^m}^\times) = \mathcal{V} \setminus \{0\}$. It is easy to check that $\proj_\mathcal{V}$ satisfies all the conditions of the retrieval function.

Let $\mathbf{M} = (m_{i}) \in \vF_{q}^{N}$ be the database, i.e., there are $N$ files in the database, each file is of size $q$. Suppose the user wants to retrieve the $b$-th file from the database. Let $\cC$ be a random  $[n,k]$ linear code over $\vF_{q^m}$.

\paragraph{Query generation:}
For the encoding and decoding, we follow the same procedure as in Section \ref{sec:pirate}.

Let $\bG$ be a generator matrix of $\C$, and let $I \subseteq \{1,\ldots,n\}$ be an information set. 
We use $\bG$ to to perform the encoding, i.e., the encoding map is $\Enc: \vF_q^k \to \vF_q^n$ given by $\ba \mapsto \ba \bG$.

Let $\ba_1,\ldots,\ba_N$ be randomly chosen vectors in $\vF_{q^m}^k$, and define the corresponding codewords
$\bc_i:= \Enc(\ba_i) = \ba_i \bG$ for all $i \in \{1, \ldots, N\}.$

As in Section \ref{sec:pirate}, we perform the decoding by adding no errors at the coordinates that are indexed by $I$.

Let $v$ be a fixed element in $I^C$. Now, we choose error vectors $\be_1,\be_2,\ldots,\be_N$ randomly in $\vF_{q^m}^n$ such that
\begin{enumerate}
    \item $\supp(\be_i) \subseteq I^C$ for all $i \in \{1, \ldots, N\} $,
    \item $\be_i[v] \in \mathcal{W}$ for all $i \neq b$,
    and $\be_b[v] \in \mathcal{V} \setminus \{0\}$.
\end{enumerate}

Let $\bq_i:=\bc_i+\be_i$ for $i \in \{1,\dots, N\}$.
The query is then given by 
\[Q := \{ \bq_1, \bq_2,\ldots,\bq_N \}.\]

\paragraph{Reply generation:} 
The response is generated by 
computing 
\[\br = \sum_{i=1}^N m_{i} \bq_i.\]

\paragraph{Reply extraction:}

Write $\br = \bc + \be$, where $\bc = \sum_{i=1}^N m_i \bc_i$ and $\be = \sum_{i=1}^N m_i \be_i$.

Since  $I$ is an information set and $\supp(\be) \subseteq I^C$, we can perform the decoding on $\br$ by computing 
$$\br - \br_I \bG_I^{-1} \bG = \be = \sum_{i=1}^N m_i \be_i.$$

Now we consider the $v$-th coordinate of $\be$ and apply the retrieval function, which gives 
\[\text{Proj}_{\mathcal{V}} \left(\sum_{i=1}^N m_{i} \be_i[v] \right) = m_b \be_b[v].\]

This works because $\be_i[v] \in \mathcal{W}$ for all $i \neq b$, and $\be_b[v] \in \mathcal{V} \setminus \{0\}$. 
Moreover, since we know $\be_b[v]$, we can retrieve $m_b$.

\subsubsection{Security}
The original PIR scheme \cite{holzbaur2020computational} has been attacked in \cite{bordage2020privacy}, by solving the distinguishability problem. The attack follows the second strategy mentioned in Section \ref{sec:general-security}. 

Let $\bA$ be the matrix containing all the query vectors as rows, i.e., $$\mathbf{A} = \begin{pmatrix} \mathbf{q}_1 \\ \mathbf{q}_2 \\ \vdots \\ \mathbf{q}_N \end{pmatrix} = \begin{pmatrix} \bc_1+\be_1 \\   \bc_2 + \be_2 \\ \vdots \\  \bc_N + \be_N \end{pmatrix} = \bC + \bE,$$ with $\bC = \begin{pmatrix} \bc_1 \\ \vdots \\ \bc_N \end{pmatrix}$ and $\bE = \begin{pmatrix} \be_1 \\ \vdots \\ \be_N \end{pmatrix}$.

For each $i \in \{1,\ldots,N\}$, let $\bA_i$ be the submatrix of $\bA$ obtained by deleting the $i$-th row. Then the $\vF_q$-rank of these matrices satisfy the following proposition. 
\begin{proposition}\cite[Proposition 3.1]{bordage2020privacy}
Let $\bA$ be given as above. Then \[\rank_{\vF_q}(\bA) = \rank_{\vF_q}(\bC) + \rank_{\vF_q}(\bE).\] Moreover, for all $i \in \{1,\ldots,N\}$ \[\rank_{\vF_q}(\bA_i) = \rank_{\vF_q}(\bC_i) + \rank_{\vF_q}(\bE_i).\] \label{prop:rank_sum}
\end{proposition}
In the case when $N<mn$, the query size becomes bigger than the size of the database, i.e., the scheme is no better than the trivial PIR protocol of downloading entire database. Hence, we assume $N \geq mn$ and we use the following corollary to distinguish the index $b$ in polynomial time. 

\begin{corollary}\cite[Corollary 3.2, Proposition 3.3]{bordage2020privacy}
Let $\bA_i$ be given as above. Then, with high probability,
\begin{enumerate}
    \item $\rank_{\vF_q}(\bA_b) = mn-s$, 
    \item for $i \neq b$, we have that $\rank_{\vF_q}(\bA_i) = mn$.
\end{enumerate}
\end{corollary}
\begin{proof}
From Proposition \ref{prop:rank_sum}, we have that $\rank_{\vF_q}(\bA_i) = \rank_{\vF_q}(\bC_i) + \rank_{\vF_q}(\bE_i)$ for all $1 \leq i \leq N$. 

In the first case, we have that $\rank_{\vF_q}(\bC_b) = mk$ and $\rank_{\vF_q}(\bE_b) = (n-k-1)m +(m-s)$ (with high probability), where the first part comes from the columns not indexed by $v$, which live in the full space $\mathbb{F}_{q^m} = \mathcal{W} + \mathcal{V}$ and the second part comes from the column indexed by $v$, which lives in the subspace $\mathcal{W}$. Note that the equation  $\rank_{\vF_q}(\bE_b) = m(n-k)-s$ holds true with high probability due to the randomness of the matrix entries.

In the case of $i \neq b$, we still have that $\rank_{\vF_q}(\bC_i) = mk$, but now $\rank_{\vF_q}(\bE_i) = (n-k-1)m +m$ (with high probability), where the first part comes from the columns not indexed by  $v$ and the second part comes from the column $v$ (observe that in this case all columns are in the full space $\mathbb{F}_{q^m} = \mathcal{W} + \mathcal{V}$). Note that the equation  $\rank_{\vF_q}(\bE_i) = m(n-k)$ holds true with high probability due to the randomness of the matrix entries.
\end{proof}

\subsection{AMG PIR scheme}  \label{sec:gaborit} 
In the following, we describe the PIR scheme presented in \cite{amg} with respect to our code-based framework. Later, we also present the lattice-based attack \cite{liu2016cryptanalysis} in terms of solving the distinguishability problem. Note that the original PIR scheme differs from our description in the following way: \begin{itemize}
    \item Database setup: in \cite{amg}, the authors consider the database elements to be vectors over the base field $\vF_p$. Moreover, each query element is a matrix over $\vF_p$. In the following description, we use an equivalent setup where the database files are single elements in $\vF_{p}$ and query elements are vectors over $\vF_p$.
    \item Noise-scrambling matrix $\Delta$: the authors introduce an invertible diagonal matrix $\Delta$ in order to disguise the soft-noise error vectors from the hard-noise error vectors. In our description, we ignore this scrambling matrix $\Delta$, as we will see in the security discussion that $\Delta$ has no effect on the column space of the query matrix.
    \item In \cite{amg}, the rate $k/n$ of the underlying linear code is fixed $k/n=0.5$. In our description we use an arbitrary rate. 
\end{itemize}

\subsubsection{Scheme} In this scheme, we work over a finite field $\vF_p$, where $p$ is a prime number. We will see $\vF_p$ as $\{-\lfloor \frac{p}{2} \rfloor, \ldots, \lfloor \frac{p}{2} \rfloor \}$.
\paragraph{Setup:}
Assume that the database is of the form $\mathbf{M} = (m_{i}) \in \{0,1,\ldots,2^\ell - 1 \}^{N}$ with $\ell = \lceil\log_2(N)\rceil+1$, i.e., there are $N$ files in the database each of size $\ell$ bits. Note that if the file size is bigger than $\ell$ bits, then we split the files in chunks of $\ell$ bits. Suppose the user wants to retrieve the $b$-th file from the database. 

Let $p$ be a prime number greater than $2^{3 \ell}$ and $t = 2^{2\ell}$. 
The retrieval function is given by the remainder of the Lee weight corresponding to modulo $t$, i.e., 
\begin{align*}
    f: \vF_p &\to \vF_p, \\ x &\mapsto x - \lweightt{x \mod t},
\end{align*} where $\mathrm{wt}_{L_t}$ denotes the Lee weight on $\vZ/t\vZ = \{0,1,\ldots,t-1 \}$, which is defined as \[\lweightt{z}:= \min\{|z|,t-|z|\}.\] The set $X = \{0,1,\ldots,2^\ell-1\}$, $Y= \{-1,1\} \subseteq \ker(f)$ and $Z= \{t\} \subseteq f^{-1}(\vF_p^\times)$. 

 Now observe that a linear combination of elements in $Y$ with scalars from $X$ having arbitrary number of terms does not necessarily belongs to $\ker(f)$. However, the condition is satisfied when we have at most $N$ number of terms in the linear combination: for $x_1,\ldots,x_N \in X$ and $y_1,\ldots,y_N \in Y$ we have that \[|x_1y_1 + \cdots + x_N y_N | \leq N 2^\ell < \frac{t}{2}, \] and hence
\begin{align*}
f\left(\sum_{i=1}^N x_iy_i\right) &= \sum_{i=1}^N x_iy_i - \lweightt{\sum_{i=1}^N x_iy_i \mod t}\\
& = \sum_{i=1}^N x_iy_i - \sum_{i=1}^N x_iy_i =0.
\end{align*}
Further we have that for $y \in Y, x \in X$ and $z \in Z$
\begin{align*}
    f(y+xz) &= f(y +xt) \\ &= y+xt - \lweightt{y+xt \mod t} \\ 
    &= y+xt - \lweightt{y \mod t} \\ &= y+xt - y  \\ &= xt = xf(t),
\end{align*}
since $f(z) = f(t)= t-\lweightt{t \mod t} = t.$

Let $\cC$ be a random linear $[n,k]$ code over $\mathbb{F}_{p}$, which is kept secret by the user.

\paragraph{Query generation:}
For the encoding and decoding, we follow the same procedure as in Section \ref{sec:pirate} and \ref{sec:lukas}.

Let $\bG$ be a generator matrix of $\C$, and let $I \subseteq \{1,\ldots,n\}$ be an information set. We use $\bG$ to perform the encoding, i.e., the encoding map is $\Enc: \vF_q^k \to \vF_q^n$ given by $\ba \mapsto \ba \bG$.

Let $\ba_1,\ldots,\ba_N$ be randomly chosen vectors in $\vF_q^k$, and define the corresponding codewords
$\bc_i:= \Enc(\ba_i) = \ba_i \bG$ for all $i \in \{1, \ldots,  N\}$.

As in Section \ref{sec:pirate} and \ref{sec:lukas}, we perform the decoding by adding no errors at the coordinates that are indexed by $I$. 

Let $v$ be a fixed element in $I^C$. Now, we choose error vectors $\be_1,\be_2,\ldots,\be_N$ randomly in $\vF_{q^m}^n$ such that
\begin{enumerate}
    \item $\supp(\be_i) \subseteq I^C$ for all $i \in \{1, \ldots, N \}$,
    \item $ \be_i[v] \in \{\pm1\}$ for all $i \neq b$, and $\be_b[v] = t$.
\end{enumerate}

Let $\bq_i:=\bc_i+\be_i$ for all $i\in \{1,\dots, N\}$.
The query is then given by 
\[Q := \{ \bq_1, \bq_2,\ldots,\bq_N \}.\]

\paragraph{Reply generation:} 
The response is generated by 
computing 
\[\br = \sum_{i=1}^N m_{i} \bq_i.\]

\paragraph{Reply extraction:}

Write $\br = \bc + \be$, where $\bc = \sum_{i=1}^N m_i \bc_i$ and $\be = \sum_{i=1}^N m_i \be_i$.

Since  $I$ is an information set and $\supp(\be) \subseteq I^C$, we can perform the decoding on $\br$ by computing 
$$\br - \br_I \bG_I^{-1} \bG = \be = \sum_{i=1}^N m_i \be_i.$$
We will only focus on the $v$-th coordinate of $\be$ and apply the retrieval function to obtain
\begin{align*}
     f\left(\sum_{i=1}^N m_i \be_i[v]\right) & =  \sum_{i=1}^N m_i \be_i[v] - \lweightt{ \sum_{i=1}^N m_i \be_i[v] \mod t} \\
     & = \left(\sum_{\substack{i=1\\i \neq b}}^N m_i \be_i[v] - \lweightt{\sum_{\substack{i=1\\i \neq b}}^N m_i \be_i[v] \mod t}\right) + m_b \be_b[v]\\
     & = m_b \be_b[v]= m_b t.
\end{align*}
 This works since \[  | \sum_{\substack{i=1\\i \neq b}}^N m_i \be_i[v]  | < t/2 \quad \mbox{and} \quad \mbox{$m_b\be_b[v]$ is a multiple of $t$}, \] and hence \[\lweightt{ \sum_{i=1}^N m_i \be_i[v] \mod t} = \lweightt{\sum_{\substack{i=1\\i \neq b}}^N m_i \be_i[v] \mod t} = \sum_{\substack{i=1\\i \neq b}}^N m_i \be_i[v].\]
 
Now since $\gcd(t,p)=1$, we can retrieve $m_b$.

\subsubsection{Security}
In \cite{liu2016cryptanalysis}, Liu et al. presented a lattice-based attack on the AMG PIR scheme. The method used in the attack can be described as per the first strategy, mentioned in Section \ref{sec:general-security}, to solve the distinguishability problem. 

Let $\bA$ be the matrix containing all the query vectors as rows, i.e., $$\mathbf{A} = \begin{pmatrix} \mathbf{q}_1 \\ \mathbf{q}_2 \\ \vdots \\ \mathbf{q}_N \end{pmatrix} = \begin{pmatrix} \bc_1+\be_1 \\   \bc_2 + \be_2 \\ \vdots \\  \bc_N + \be_N \end{pmatrix}.$$
As discussed in the security part of Section \ref{sec:pirate}, the vector $\left(\be_1[v],\be_2[v],\ldots,\be_N[v]\right)$ belongs to the column span of $\bA$. 

Recall that by construction, the vector $\left(\be_1[v],\be_2[v],\ldots,\be_N[v]\right)$ has $N-1$ entries from $\{-1,+1\}$ and one entry with value equal to $t$. If we delete the $b$-th row of $\bA$, call it the matrix $\bA_b$, then the vector $\left(\be_1[v],\ldots,\be_{b-1}[v], \be_{b+1}[v],\ldots,\be_N[v]\right)$ will be, with a very high probability, the shortest vector in the $p$-ary lattice generated by the columns of $\bA_b$. More precisely, the lattice is generated by the $n$ columns of $[\bA_b|p\mathbf{Id}_{N-1}]$. However, it is still infeasible to find this vector due to the large dimension of the lattice.

In \cite{liu2016cryptanalysis}, the authors construct multiple small dimensional lattices. Let $k \leq s \leq N$, and let $\bA^{(1)},\ldots,\bA^{(\lceil N/s \rceil)}$ be a row-wise partitioning of the matrix $\bA$, i.e.,  $\bA^{(i)}$ is the $s \times n$ matrix given by $s$ rows of $\bA$ indexed by $\{(i-1)s+1,\ldots,is\}$. Now, let $\mathcal{L}_i$ be the $p$-ary lattice generated by the columns of $\bA^{(i)}$. Note that the dimension of the lattices $\mathcal{L}_i$ is $s$, hence the attacker chooses $s$ such that implementing basis reduction algorithms for $\mathcal{L}_i$ is feasible. In order to find the index $b$, the attacker goes through each of these lattices. 

Note that the index $b$ of the desired file corresponds to the lattice $\mathcal{L}_{\lfloor b/s \rfloor}$, which the attacker is able to find,
and then the attacker finds the index $b$ by solving the closest vector problem for $\mathcal{L}_{\lfloor b/s \rfloor}$. 

More in detail, in the case of $i \neq \lfloor b/s\rfloor$, we observe that the shortest vector in $\mathcal{L}_i$ corresponds to the vector $(\be_{(i-1)s+1}[v],\ldots,\be_{is}[v])$ having entries in $\{-1,+1\}$. This observation does not hold in the case of $i=\lfloor b/s \rfloor$ due to the existence of large $t$. The attacker uses the lattice reduction algorithms to find the shortest vector in each $\mathcal{L}_i$, and consequently finds the corresponding lattice $\mathcal{L}_{\lfloor b/s \rfloor}$. 

Now, the index $b$ can be located using solving the closest vector problem. Let $j = \lfloor b/s \rfloor$. Then observe that $(\be_{(j-1)s+1}[v],\ldots,\be_{js}[v]) \in \mathcal{L}_j$ is the closest lattice vector to $(0,\ldots,0,t,0,\ldots,0)$ (with $t$ at the $b$-th position). To find the index $b$, we can use  Kannan's embedding technique \cite{kannan1987minkowski} to solve (at most) $s$ instances of the closest vector problem with inputs vector of the form $(0,\ldots,0,t,0,\ldots,0)$.

\subsection{Ring-LWE based PIR schemes} \label{sec:lwe} 
In the section, we describe the PIR schemes constructed using the Ring-LWE (RLWE) based homomorphic encryption schemes. In particular, we consider the construction of XPIR scheme \cite{XPIR} that uses the Ring-LWE based homomorphic encryption scheme presented in \cite{brakerski2011fully}. 

The original PIR scheme differs from our description in the error distribution as follows.
In \cite{XPIR}, the authors use two different distributions $\chi$ and $\chi^\prime$ to sample errors. The distribution $\chi$ is used to generate the public key and the distribution $\chi^\prime$, having larger variance, is used for encryption. In the following description, we consider only one distribution, mimicking $\chi^\prime$, to sample error vectors in the query generation process.

We would like to remark that in the following description, the database elements and the query elements are polynomials of degree smaller than $n$ with coefficients in $\R$, which can also be represented by vectors in $\R^n$.

\subsubsection{Scheme}
In this scheme, we work over a finite ring $\vZ/q\vZ$, where $q$ is a positive integer. Instead of a random linear code over $\vZ/q\vZ$, we consider a random constacyclic code over $\vZ/q\vZ$.
\paragraph{Setup:}
Let $q, t$ be positive integers with $t<q$ and $\gcd(t,q)=1$. The retrieval function is given by
\begin{align*}
  f: \vZ/q\vZ &\to \vZ/q\vZ , \\
x &\mapsto x \pmod{t}.
\end{align*}
Let $\chi$ be a discrete Gaussian distribution with standard deviation $\sigma$. The parameters $q,n,t,\sigma$ are chosen such that they satisfy $Nt^2\sigma \sqrt{n} < q/2$, where $n$ is the length of the linear code that will be used in query generation. \\
Now, we define the subsets 
\begin{align*}
&X = \{0,\ldots,t-1\} \subseteq \vZ/q\vZ,\\
&Y= \{t y \mid y \mbox{ is sampled from the distribution } \chi \},\\
&Z = \{ t y +1 \mid y \mbox{ is sampled from the distribution } \chi\}.
\end{align*}
Observe that for $x_1,x_2,\ldots,x_N \in X $ and $ty_1,ty_2,\ldots,ty_N \in Y$ we have that
\begin{align*}
    f\left(\sum_{i=1}^N x_i ty_i\right) &= \sum_{i=1}^N x_i t y_i  \mod t \\
    &= 0.
\end{align*}
This works since the choice of parameters $q,n,t,\sigma$ implies that $|\sum_{i=1}^N x_i t y_i| < q/2$ with very high probability. 
And for $x \in X, ty \in Y$ and $tz+1 \in Z$ we have that 
\begin{align*}
    f(y+xz) &= ty + x(tz+1)   \mod t \\
    &= x \mod t = x = x f(z),
\end{align*}
since $|ty+x(tz+1)| < q/2$.

Let $n$ be a power of 2, and let $R_q := (\vZ/q\vZ)[x]/(x^n+1)$. 
Let $\mathbf{M} = (m_{i}) \in \left(X[x]/(x^n+1)\right)^N$, i.e., there are $N$ files in the database and each file is an element in $R_q$ with coefficients in $X$. In particular, each file is of size $\log_2(tn)$ bits. Suppose the user wants to retrieve the $b$-th file from the database. \\
Let $\cC$ be a constacyclic code of length $n$ over $\vZ/q\vZ$ generated by some randomly chosen $s \in R_q$, i.e., $\cC$ is a ideal in $R_{q}$ generated by $s$. The code is kept secret by the user.

 \paragraph{Query generation:}
We use the generating polynomial $s$ to define the encoding map, i.e., $Enc : R_q \to R_q$ is given by $a \mapsto as$.

Let $a_1,a_2,\ldots,a_N$ be randomly chosen elements in $R_q$, and define $N$ codewords $c_i := a_is$ for all $i \in \{1,\ldots,N\}$.

 Now, we choose the errors $e_1,e_2,\ldots,e_N$ in $R_{q}$ such that they satisfy the following two conditions that allow the reply extraction: 
 \begin{enumerate}
    \item $e_i = t y_i$, with $y_i$ sampled from the distribution $\chi$, for all $i \neq b$,
     \item $e_b = t y_b + 1$ with  $y_b$ sampled from $\chi$. 
 \end{enumerate}

Let $\bq_i:= (a_i,c_i + e_i)$ for all $i\in \{1,\dots, N\}$.
 The query is then given by 
 \[Q := \{ \bq_1, \bq_2,\ldots,\bq_N \}.\]

 \paragraph{Reply generation:} 
 The response is generated by 
 computing 
 \[\br = \sum_{i=1}^N m_{i} \bq_i = \sum_{i=1}^N (m_ia_i, m_ic_i + m_ie_i)=: (r_1, r_2).\]

 \paragraph{Reply extraction:}

By applying the encoding map $Enc$ on $r_1$, we first decode $r_2$ to obtain the error part, i.e., 
\[r_2 - Enc(r_1) = r_2 - s r_1 = \sum_{i=1}^N m_i e_i.\] 
  After that we can use the retrieval function $f$, 
\begin{align*}
    f\left(\sum_i^N m_i e_i \right) &= \sum_{i=1}^N m_i t y_i + m_b \pmod{t} \\
    & = m_b.
\end{align*}
Note that here we apply $f$ on an element of $R_q$, which is done by applying $f$ on each coefficient.

The last equality follows from the conditions on the parameters $n,q,t,\sigma$, since the maximal coefficient of $\sum_{i=1}^N m_i e_i$ is, with high probability, upper bounded by $Nt^2\sigma \sqrt{n}$ (see \cite[Lemma 1]{brakerski2011fully}), which is less than $q/2$.

\subsubsection{Security}
The security of this scheme is based on the hardness of solving the polynomial learning with error (PLWE) problem, which is a simplified version of the ring LWE problem. 

Let $R_q = \vZ/q\vZ[x]/(x^n+1)$, and let $\chi$ be a narrow discrete Gaussian distribution on $R_q$. Then the PLWE assumption states that it is computationally hard to distinguish a polynomial number of samples of the form $(a_i,a_is+e_i)$ and the same number of samples of the form $(a_i,u_i)$, where $s, a_i$'s and $u_i$'s are sampled uniformly from $R_q$ and the $e_i$'s are sampled from $\chi$.

Moreover, \cite[Proposition 1]{brakerski2011fully} states that if the samples are of the form $(a_i,a_is+te_i)$, where $a_i,s,e_i$ are as above and $t \in (\vZ/q\vZ)^\times$, then distinguishing such samples from the uniform samples is equivalent to the PLWE assumption. 

\renewcommand{\arraystretch}{1.7}
\afterpage{\begin{landscape}
\begin{table}[p]
    \centering
    \begin{tabular}{|p{3cm}|K{3cm}|K{4.7cm}|c|K{2cm}|c|}  
    \hline {\centering \textbf{PIR Scheme}} & \textbf{Base ring $\R$}  & \textbf{Retrieval function} & \textbf{Set} $X$ & \textbf{Set} $Y$ & \textbf{Set} $Z$ \\ \hline
    Basic PIR using field homomorphisms     & a finite field $\vF_q$ & $\begin{aligned}[t]
         f:  \vF_q & \to \vF_q \\
           x & \mapsto x
    \end{aligned}$ & $\vF_q$ & $\{0\}$ & $\vF_q^\times$ \\ \hline
    HHWZ PIR & a finite field extension $\vF_{q^m}$  & $\begin{aligned}[t]
         f:  \vF_{q^m} & \to \vF_{q^m} \\
           x & \mapsto \proj_{\mathcal{V}}(x),
    \end{aligned}$ \linebreak where $\mathcal{V}$ is a non-trivial $\vF_q$-subspace of $\vF_{q^m}$  & $\vF_q$ & $\mathcal{W}$ \linebreak such that $\mathcal{V} \oplus \mathcal{W} = \vF_{q^m}$ & $\mathcal{V} \setminus \{0\}$ \\ \hline
    AMG PIR & a prime field $\vF_p$ where $p$ is prime greater than $2^{3\ell}$ for some positive integer $\ell$ & $\begin{aligned}[t]
    f: \vF_p &\to \vF_p, \\ x &\mapsto x - \lweightt{x \mod t},
\end{aligned}$ where $t = 2^{2\ell}$  & $\{0,1,\ldots,2^\ell-1\}$ & $\{-1,1\}$ & $\{t\}$ \\  \hline
    LWE-based PIR & $\vZ/q\vZ$ \linebreak for some positive integer $q$ & $\begin{aligned}[t]
         f:  \vZ/q\vZ & \to \vZ/q\vZ \\
           x & \mapsto x \pmod{t},
    \end{aligned}$ \linebreak for some positive integer $t <q $ with $\gcd(t,q)=1$ & $\{0,\ldots,t-1\}$ & $\{t y \mid y \leftarrow \chi \}$ & $\{ t y +1 \mid y \leftarrow  \chi\}$\\ \hline
    \end{tabular}
    \caption{Comparison of different PIR schemes with respect to the code-based framework.}
    \label{comparison}
\end{table}
\end{landscape}}
\renewcommand*{\arraystretch}{1}

\section{Theoretical remarks} \label{sec:remarks} \subsection{Generic PIR scheme vs code-based framework}

A natural question would be to ask whether any single server PIR scheme can be described in terms of the code-based framework. The answer is no, as the number theoretic PIR scheme by Kushilevitz and Ostrovsky \cite{kushilevitz1997replication} does not fit the framework. However, if we restrict to the class of PIR schemes that generates replies by contracting the database elements and the query elements using linear combinations (which will be denoted from now on as additive PIR schemes), then the answer is yes.  In the following, we discuss the requirements of an arbitrary additive PIR scheme and argue the necessity of the elements in the code-based framework to fulfil those requirements:
\begin{enumerate}
    \item \textit{Ambient space:} An additive PIR scheme needs two operations: multiplication $(\ast)$ between database elements and query elements, and addition $(+)$ of those products. Hence, the canonical choice of the ambient space is rings. For practical reasons, the rings should be finite. 
    \item \textit{Retrieval:} Let the database be denoted by $\mathcal{DB}=\{db_1,\ldots,db_N\}$, and the corresponding query be given by $Q = \{q_1,\ldots,q_N\}$. Suppose that the user wants to retrieve the $b$-th file. In an additive PIR scheme, the reply is $\sum_{i=1}^N db_i \ast q_i$ and user wants to retrieve $db_b$ from the reply. The operation $\sum_{i=1}^N db_i \ast q_i \mapsto db_b$, denoted by $g$, is an analogue to the retrieval function used in the code-based framework. First we note that $g$ annihilates $\sum_{i\neq b} db_i \ast q_i$ in such a way that we are only left with $g(db_b \ast q_b)$. And then $db_b$ is recovered from $g(db_b \ast q_b)$. These two properties imply that $db_i$'s and $q_i$'s live in special subsets of the ambient space $R$. Let $X$ denote the space of database elements, $Y$ denote the space of query elements that are not associated with the desired file and $Z$ denote the space of query element associated with the desired file. The requirements on $g$ imply that: (1) a linear combination of elements in $Y$ with scalars in $X$ belongs to the kernel of $g$, and (2) $g(x \ast z) = x \ast g(z)$ and $g(z)$ is an invertible element, for any $x \in X$ and $z \in Z$. These two conditions are the basis of the conditions of the retrieval function used in the code-based framework.
    \item \textit{Privacy:} Another important aspect of a PIR scheme is privacy, i.e., given a query $Q$, it should be computationally infeasible to determine the index $b$ of the desired file. Let us look at the scenario where we directly use elements in $Y$ and $Z$ to generate query elements. Then the privacy relies on the hardness of the following decisional problem: given $q \in Y \cup Z$, decide whether $q \in Y$ or $q \in Z$. In general this may not be a hard problem, as one can apply the retrieval function to distinguish the elements between $Y$ and $Z$. Therefore, to ensure privacy we must add some randomness to the query elements. Moreover, the user should be able to remove this randomness even after receiving the reply that contains their linear combinations. This is exactly the rationale of linear error-correcting codes. We treat the elements of $Y$ and $Z$ as errors, and the added randomness belongs to a random linear code. 
\end{enumerate}

\subsection{On security of PIR schemes}
In terms of the code-based framework, the security of a PIR scheme relies on the type of the underlying retrieval function. As we have noticed from the examples in Section \ref{sec:examples}, the following type of retrieval functions are not safe to use.
\begin{enumerate}
    \item \textit{Field homomorphisms:} In the case where the retrieval function is a non-trivial field homomorphism, the PIR scheme is then equivalent to the one described in Section \ref{sec:pirate}. The kernel of the retrieval function must be $\{0\}$, as $\{0\}$ is the only proper ideal in any field. As a consequence, determining the index of the desired file becomes an easy task of finding a unitary vector in the column space of the query matrix, thus it suffers from the first attack strategy discussed in Section \ref{sec:general-security}. 
    \item \textit{Vector space homomorphisms:} In this case, the resulting PIR scheme is equivalent to HHWZ PIR scheme \cite{holzbaur2020computational}, described in Section \ref{sec:lukas}. The kernel of a non-trivial linear map is a proper subspace of the parent vector space. This results in an exceptionally low rank of the matrix that is obtained from the query matrix by deleting the row that corresponds to the desired file, thus it suffers from the second attack strategy discussed in Section \ref{sec:general-security}.
\end{enumerate}

We can generalize these two cases to more types of retrieval functions. 
Clearly, the weakness of vector space homomorphisms can also be observed in the case of free module homomorphims, because of the existence of the notion of rank and dimension for free modules. On the other hand, the weakness of field homomorphisms can be seen in the case of local ring homomorphims. Let $R$ be a finite local ring with maximal ideal $M$, then the kernel of the retrieval function is a subideal of $M$. Note that there exists an integer $\ell$ such that $M^\ell = \{0\}$ and $M^{\ell-1} \neq\{0\}$. Let $a \in M^{\ell-1} \setminus \{0\}$, then note that $ar = 0$ for all $r \in M$. This implies that the special column vector $(\be_1[v],\ldots,\be_N[v])$, when multiplied by $a$, results in a unit vector. Hence, similar to the field homomorphism case, we observe the existence of a unit vector in the column space of the query matrix. 

The other two schemes, presented in Section \ref{sec:gaborit} and \ref{sec:lwe} respectively, do not use additive retrieval functions. Both the schemes work on the idea of using small modulus errors in a large modulus ambient space. Due to which the security eventually relies on finding short vectors in a high dimensional lattice, which is a computationally hard problem. However, in the case of AMG PIR scheme, the problem breaks down over multiple small dimensional lattices and hence the attack becomes feasible. 

In order to construct an additive PIR scheme, one may investigate the cases of structured morphisms like ring homomorphisms and module homomorphisms, or the cases of unstructured morphisms like the functions used in AMG scheme and LWE-based schemes. 

Furthermore, if one constructs an additive PIR scheme independently, then it would be worth translating the scheme in terms of the code-based framework to check for possible security issues.

\section*{Acknowledgments}
The authors would like to thank Lukas Holzbaur, Antonia Wachter-Zeh and Camilla Hollanti for useful discussions and Razane Tajeddine for bringing this interesting topic to their knowledge. 
This work was partially supported by Swiss National Science Foundation grant no. 188430 and Forschungskredit of the University of Zurich grant no. FK-19-080.

\bibliographystyle{plain}
\bibliography{references}

\end{document}